\documentclass[twocolumn]{IEEEtran}
\usepackage{amsmath}
\usepackage{amsthm}
\usepackage{amsfonts}
\usepackage{amssymb}
\usepackage{mathrsfs}
\usepackage{cases}
\usepackage{cite}
\usepackage{graphicx}
\usepackage{mathrsfs}
\usepackage{subfigure}
\usepackage{epstopdf}
\usepackage{color}
\usepackage{bm}
\usepackage{caption}
\allowdisplaybreaks[4]
\theoremstyle{remark}
\newtheorem{theorem}{\hspace{1em}Theorem}
\newtheorem{lemma}{\hspace{1em}Lemma}

\begin{document}

\title{\LARGE Combating the Control Signal Spoofing Attack in UAV Systems
\thanks{The authors are with the
School of Electronics and Information Engineering, and also with the Ministry
of Education Key Lab for Intelligent Networks and Network Security, Xi'an Jiaotong
University, Xi'an, 710049, Shaanxi, P. R. China. Email: {\tt
xjtu-huangkw@outlook.com,
xjbswhm@gmail.com}.
}
\author{Ke-Wen Huang, \hspace{0.05in} Hui-Ming Wang,~\IEEEmembership{Senior Member,~IEEE}
}
}
\maketitle

\begin{abstract}
     Unmanned aerial vehicle (UAV) system is vulnerable to the control signal spoofing attack due to the openness of the wireless communications. In this correspondence, a physical layer approach is proposed to combat the control signal spoofing attack, i.e,. to determine whether the received control signal packet is from the ground control station (GCS) or a potential malicious attacker (MA), which does not need to share any secret key. We consider the worst case where the UAV does not have any prior knowledge about the MA. Utilizing the channel feature of the angles of arrival, the distance-based path loss, and the Rician-$\kappa$ factor,  we construct a generalized log-likelihood radio (GLLR) test framework to handle the problem. Accurate approximations of the false alarm and successful detection rate are provided to efficiently evaluate the performance.
\end{abstract}

\begin{IEEEkeywords}
Physical layer authentication, spoofing attack, UAV system, generalized likelihood radio,  false alarm rate.
\end{IEEEkeywords}

\section{Introduction}
\label{Sec:Introduction}
Recent years has witnessed the rapid development of the unmanned aerial vehicle (UAV) technique.
Thanks to the high flexibility of the UAV, it finds widespread applications in both civilian and military fields \cite{Vachtsevanos2015Handbook}.
In general, the UAV will connect with the ground control station (GCS) through wireless links for data and control signal exchanges.
However, due to the openness of the wireless environments, the UAV systems are vulnerable to the control signal spoofing attacks \cite{Y.Zeng2016CM}.
More specifically, a malicious attacker (MA) can forge the control signal of the GCS to take over the UAV  illegitimately, which poses a serious threat to the safety of the UAV systems.
Therefore, identifying the source of the received signal at the UAV-side (whether the GCS or any other MA) is of great significance in the UAV systems, which motivates this work.

In this correspondence, we propose to utilize the characteristics of the physical layer channels to combat the control signal spoofing attack.
{
We note that the characteristics of the physical layer channels have already been used to cope with the eavesdropping attack in both the terrestrial communication systems, e.g., in \cite{H.Zhang2016TII}, and the UAV communication systems, e.g., in \cite{Liu2017WCSP,QWang2017ComLett,H.Liu2017ICUFN}, but have not been specifically considered to combat the control signal spoofing attack in the UAV systems.
}
The basic idea of the proposed method in this correspondence comes from the fact that the wireless channels from the GCS and a MA will be different significantly due to the prominent distinction of propagation environments. The MA is not able to imitate the channel characteristics of the GCS.
In general, using the parameters of the physical channel to identify the signal source falls into the problem of physical layer authentication (PLA), which has attracted widespread attention \cite{XiaoLiang2008TWC,P.Baracca2012TWC,J.K.Tugnait2013JSAC,W.Hou2014TCOM,
	J.Liu2016TWC}.
In these works, various channel parameters, such as channel coefficients, power spectral, frequency offsets, and multi-path delays are exploited to perform the PLA.
However, these methods are not specifically suitable for the UAV ground-to-air authentication (GAA). In general, the widely used Rayleigh fading channel model is not applicable to the ground-to-air channel because the scatterers around the UAV is generally sparse, and the ground-to-air channel usually contains a strong line-of-sight (LOS) component. Furthermore, the UAV is under control so some parameters relative to the trajectory will be previously known. None of these UAV-special features has been taken into account in above works. To the best of our knowledge, using the physical layer approach for control signal authentication in a UAV system does not appear in existing literature.

In this correspondence, we design a physical-layer-based method for the UAV to identify the source of the received  signal to combat the control signal spoofing attack.
Utilizing the angles of arrival, the distance-based path loss, and the Rician-$\kappa$ factor of the channel,  we construct a generalized log-likelihood radio (GLLR) test framework to recognize whether the signal comes form the legitimate GCS or the illegitimate MA.
The contributions can be summarized as follows:
1) the GAA problem is considered under the worst case where no prior knowledge of the MA is obtained at the UAV-side; 2) we propose a generalized log-likelihood radio (GLLR) test method to handle the GAA problem, wherein the unknown channel parameters are first estimated; and 3) we perform comprehensive analysis to evaluate the authentication performance, where  accurate approximations for the false alarm rate (FAR) and successful detection rate (SDR) are derived.

\emph{Notations:} $(\cdot)^T$ and $(\cdot)^H$ denote the transpose and conjugate transpose, respectively. $\mathcal{C}^{L}$ denotes the space of $L$-dimensional column vector. The complex Gaussian distribution is denoted by $\mathcal{CN}\left(\bm{m},\bm{R}\right)$ with $\bm{m}$ and $\bm{R}$ being the mean vector and covariance matrix, respectively. $\mathcal{G}\left(a,b\right)$ means the Gamma distribution with $a$ and $b$ denoting the shape and scale parameters, respectively.  $\bm{v}_{(i,j)}$, for $j>i$, denotes a vector which is comprised of the $i^{\mathrm{th}},(i+1)^{\mathrm{th}},\cdots,j^{\mathrm{th}}$ elements in vector $\bm{v}$.

\section{System Model}
\label{Sec:SystemModel}
A comprehensive system model is provided in Fig. \ref{Fig:Model}, which consists of a multi-antenna UAV with $L$ antennas, a single antenna GCS, and a single antenna MA. The UAV is under the control of the GCS and the MA performs control signal spoofing attack by sending a deceitful control signal which follows exact the same format of the real control signal. {We note that in this correspondence, we only consider a single MA for simplicity. The cases with multiple MAs are generally complicated and thus left for future research.}
For simplicity and clarity, all the parameters in this paper that are labeled with subscript $0$ ($1$) means that they are associated with the GCS (MA).

For $i\in\{0,1\}$, assume that the azimuth and zenith angles of arrival of the control signal   at the UAV are denoted by $\left(\theta_i,\phi_i\right)$, then the ground-to-air channels are given by
$\bm{f}_i\triangleq d_i^{-\alpha/2}\left(\sqrt{1/(1+\kappa_i)}\bm{a}_i +
\sqrt{\kappa_i/(1+\kappa_i)}\bm{h}_i\right)\in \mathcal{C}^{L}$, where  $\alpha$ is the exponential path loss factor, $d_i$ denotes the distance, $\kappa_i$ is the Rician-$\kappa$ factor,
$\bm{a}_i\triangleq\bm{a}\left(\omega_i,\mu_i\right)$, satisfying $\left\|\bm{a}_i\right\|^2 = L$, is the steering vector of the antenna array \footnote{The specific form of the steering vector, i.e., $\bm{a}(\omega,\mu)$, is provided in simulation part in Section IV.} with $\omega_i\triangleq\sin\left(\theta_i\right)\cos\left(\phi_i\right)$ and $\mu_i\triangleq\sin\left(\theta_i\right)\sin\left(\phi_i\right)$, which composes the LOS component of the channel, and $\bm{h}_i$ follows $\mathcal{CN}\left(\bm{0},\bm{I}_L\right)$ which composes the non-LOS component of the channel. For notational simplicity, we define $\lambda_i\triangleq\sqrt{1/(1+\kappa_i)}$ and $\delta_i\triangleq\sqrt{1 - \lambda_i^2}$ \footnote{The system model can be further extended to the cases where the GCS and the MA are equipped with multiple antennas. In fact, after performing the beamforming, the multi-antenna GCS and the multi-antenna MA degrade to a single-antenna GCS and a single-antenna MA, respectively. As for the UAV, our model requires that the UAV is equipped with a two-dimensional array to insure that the physical channel is a function of both the azimuth and zenith angles of arrival.}.

Since the UAV is under the control of the GCS and the trajectory will be known a priori,  we assume the UAV has the prior knowledge of the GCS at some instance, i.e., $\{d_0,\lambda_0,\omega_0,\mu_0\}$, but does not know $\{d_1,\lambda_1,\omega_1,\mu_1\}$ of the MA.
The GAA problem here is that after receiving a control signal packet, the UAV need to determine whether this packet is sent by the GCS based on its prior knowledge on $\{d_0,\lambda_0,\omega_0,\mu_0\}$.
\begin{figure}[!t]
\begin{center}
\includegraphics[width=2.5 in]{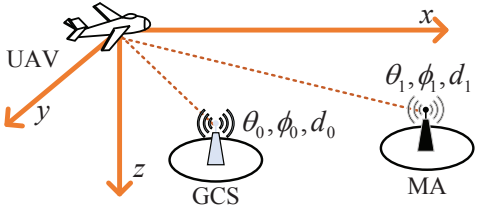}
\caption{System Model.}\label{Fig:Model}
\end{center}
\vspace{-7mm}
\end{figure}

\section{Physical Layer Authentication Scheme}
\label{Sec:PLAScheme}
We provide our PLA scheme in this section.
Assume each received packet includes a previously known training sequence, denoted by $\bm{s}\in\mathcal{C}^{L_s}$, where $L_s$ is the sequence length.
In general, the training sequence is used to estimate the instantaneous channel, which can be further used to identify the signal source because different locations have different channels.
Without loss of generality, we assume that $\left\|\bm{s}\right\|^2=1$.
The received training sequence at the UAV is written as
\begin{align}
\bm{Y} \triangleq \left\{
\begin{aligned}
&\sqrt{P_0}\bm{f}_0\bm{s}^H+\bm{N}, &&\mathcal{H}_0,\\
&\sqrt{P_1}\bm{f}_1\bm{s}^He^{\mathrm{j}\psi}+\bm{N}, &&\mathcal{H}_1,
\end{aligned}
\right.\label{OriginalReceivedSignal}
\end{align}
where $\mathcal{H}_0$ ($\mathcal{H}_1$) denotes the hypothesis that the packet is sent by the GCS (MA), $P_i$ denotes the transmit power, $\bm{N}$ is the noise with each of its element following $\mathcal{CN}\left(0,\sigma_{\bm{N}}^2\right)$, and $\psi$ is a phase rotation set by the MA which is unknown by the UAV.
Matching $\bm{Y}$ with $\bm{s}$, we obtain
\begin{align}
\bm{y} \triangleq \bm{Y}\bm{s} = \left\{
\begin{aligned}
&\sqrt{P_0}d_0^{-\alpha/2}\left(\lambda_0\bm{a}_0 + \delta_0\bm{h}_0 \right)+\bm{n}, &&\mathcal{H}_0,\\
&\sqrt{P_1}d_1^{-\alpha/2}\left(\lambda_1\bm{a}_1 + \delta_1\bm{h}_1 \right)e^{\mathrm{j}\psi}+\bm{n}, &&\mathcal{H}_1,
\end{aligned}
\right.\nonumber
\end{align}
where we have $\bm{n}\triangleq\bm{N}\bm{s}\sim\mathcal{CN}\left(\bm{0},\sigma^2\bm{I}_L\right)$ with $\sigma^2=\sigma_{\bm{N}}^2/L_s$.
We assume here that $P_0$ is known by the UAV but $P_1$ is not.
Once obtain $\bm{y}$, we propose to check the GLLR to determine the source of the received signal, either the GCS or the MA.
The GLLR test is given by
\begin{align}
T &\triangleq \frac{1}{L}\ln\frac{\max\limits_{\mathcal{I}_1,\psi}f\left(\bm{y}|\mathcal{H}_1,\mathcal{I}_1,\psi\right)}
{f\left(\bm{y}|\mathcal{H}_0,\mathcal{I}_0\right)}\gtreqless_{\mathcal{H}_0}^{\mathcal{H}_1} \tau, \label{GLLRInitial}
\end{align}
where $\mathcal{I}_i = \{P_i, d_i, \lambda_i, \omega_i,\mu_i\}$ for $i \in \{0,1\}$, and $\tau$ is the decision threshold.

As we can see from \eqref{GLLRInitial}, the basic steps of the GLLR test is that we first estimate the unknown $\{P_1, d_1, \lambda_1, \omega_1,\mu_1,\psi\}$ from $\bm{y}$ according to the maximum likelihood criterion by assuming that $\mathcal{H}_1$ is true. Then, the probability density function (PDF) of $\bm{y}$ conditioned on $\mathcal{H}_0$ is compared with the maximized PDF of $\bm{y}$ condition on $\mathcal{H}_1$ through the log-likelihood function. The final decision between $\mathcal{H}_0$ and $\mathcal{H}_1$ depends on the obtained log-likelihood radio $T$ and the decision threshold $\tau$.

In the following part of this section, we first provide the estimation result of $\{P_1, d_1, \lambda_1, \omega_1,\mu_1,\psi\}$. After that, a semi-closed form approximation of the FAR are proposed to efficiently determine the decision threshold $\tau$. Finally, we provide an approximation of the SDR to evaluate the performance of the proposed authentication method.

\subsection{Parameter Estimation}
In this subsection, we derive the estimation result of the unknown $\{P_1, d_1, \lambda_1, \omega_1,\mu_1,\psi\}$ in \eqref{GLLRInitial}, which is an indispensable step to check the GLLR.

If $\mathcal{H}_1$ is true, $\bm{y}$ follows $\mathcal{CN}\left(x_1\lambda_1\bm{a}_1e^{\mathrm{j\psi}},\epsilon_1^2\bm{I}_L\right)$, where
$x_1\triangleq \sqrt{P_1}d_1^{-\alpha/2}$.
We observe that $P_1$ is always coupled with $d_1$ in the form of $\sqrt{P_1}d_1^{-\alpha/2}$. Therefore it is sufficient to only estimate $x_1$, instead of estimating both $P_1$ and $d_1$.
The maximum likelihood estimation of $\left\{x_1, \lambda_1, \omega_1,\mu_1,\psi\right\}$ can be written as a maximization problem given by
\begin{align}
&\{x_1^*, \lambda_1^*, \omega_1^*,\mu_1^*,\psi^*\} =
\mathop\mathrm{argmax}\limits_{x, \lambda, \omega,\mu,\psi} \ln f\left(\bm{y}|\mathcal{H}_1,x, \lambda, \omega,\mu,\psi\right)
\label{OrignalProblem}
\end{align}
where $\bar{\bm{y}} \triangleq \bm{y}/\sqrt{L}$ and $\bar{\bm{a}}\left(\omega,\mu\right)\triangleq\bm{a}\left(\omega,\mu\right)/\sqrt{L}$. The estimation results are provided in the following theorem.

\begin{theorem}
The solution of \eqref{OrignalProblem} is given as follows
\begin{subequations}
\label{ParameterEstimation}
\begin{align}
&\left(\omega_1^*,\mu_1^*\right)
= \mathop\mathrm{argmax}\limits_{\omega,\mu} \quad \left|\bar{\bm{y}}^H\bar{\bm{a}}\left(\omega,\mu\right)\right|^2, \label{OmigaMuEstimation}\\
&x_1^*=\left\{
\begin{matrix}
\left\|\bar{\bm{y}}\right\|^2-\sigma^2,&
\text{if }\Xi^{*}\geq0,\\
\left|\bar{\bm{y}}^H\bar{\bm{a}}\left(\omega_1^*,\mu_1^*\right)\right|^2,&\text{if }\Xi^{*}<0,\\
\end{matrix}\right.\\
&\left(\lambda_1^*\right)^2=\left\{
\begin{matrix}
\left(\left\|\bar{\bm{y}}\right\|^2-\sigma^2\right)/\left|\bar{\bm{y}}^H
\bar{\bm{a}}\left(\omega_1^*,\mu_1^*\right)\right|^2,
\text{if }\Xi^{*}\geq0,\\
\quad 1,\quad\quad\quad\quad\quad\quad\quad\quad\quad\quad\quad\ \ \text{if }\Xi^{*}<0,\\
\end{matrix}\right.\\
&e^{-\mathrm{j}\psi^*}=
\bar{\bm{y}}^H\bar{\bm{a}}\left(\omega_1^*,\mu_1^*\right)/
\left|\bar{\bm{y}}^H\bar{\bm{a}}\left(\omega_1^*,\mu_1^*\right)\right|,
\end{align}
\end{subequations}
where $\bar{\bm{y}} \triangleq \bm{y}/\sqrt{L}$, $\bar{\bm{a}}\left(\omega,\mu\right)\triangleq\bm{a}\left(\omega,\mu\right)/\sqrt{L}$, and $\Xi^*\triangleq\left\|\bar{\bm{y}}\right\|^2-\left|\bar{\bm{y}}^H
\bar{\bm{a}}\left(\omega_1^*,\mu_1^*\right)\right|^2-\sigma^2$ and the value of $\left(\omega_1^*,\mu_1^*\right)$ can be extensively searched within the region of $\left(-1,1\right)\times\left(-1,1\right)$ in a two-dimensional real space.
\end{theorem}
\begin{proof}
To solve \eqref{OrignalProblem}, our basic idea is that we first derive the expression of $x_1, \lambda_1, \psi$ with respect to $\left(\omega,\mu\right)$, denoted by $x_1\left(\omega,\mu\right)$, $\lambda_1\left(\omega,\mu\right)$, and $\psi\left(\omega,\mu\right)$. Then, we only need to search for $\left(\omega_1^*,\mu_1^*\right)$. Based on \eqref{OrignalProblem}, we first solve
\begin{align}
&\left(x_1\left(\omega,\mu\right), \lambda_1\left(\omega,\mu\right), \psi\left(\omega,\mu\right)\right)
\nonumber \\
&=
\mathop\mathrm{argmin}\limits_{x, \lambda,\psi} \left\{
\begin{aligned}
&\frac{\left\|\bar{\bm{y}}-x\lambda\bar{\bm{a}}\left(\omega,\mu\right)e^{\mathrm{j}\psi}\right\|^2}
{x^2\left(1-\lambda^2\right)+\sigma^2}\\
&\quad\quad +\ln\left(x^2\left(1-\lambda^2\right)+\sigma^2\right)
\end{aligned} \right\}.\label{Step1}
\end{align}
Obviously from \eqref{Step1}, we obtain that
\begin{align}
&\psi\left(\omega,\mu\right) =
\mathop\mathrm{argmin}\limits_{\psi} \quad \left\|\bar{\bm{y}}-x\lambda\bar{\bm{a}}\left(\omega,\mu\right)e^{\mathrm{j}\psi}\right\|^2\nonumber \\
&\Rightarrow
e^{-\mathrm{j}\psi\left(\omega,\mu\right)}=\bar{\bm{y}}^H\bar{\bm{a}}\left(\omega,\mu\right)/
\left|\bar{\bm{y}}^H\bar{\bm{a}}\left(\omega,\mu\right)\right|.
\label{Estimate:psi}
\end{align}
Inserting \eqref{Estimate:psi} into \eqref{Step1}, we have
\begin{align}
&\left(x_1\left(\omega,\mu\right),\lambda_1\left(\omega,\mu\right)\right) \nonumber\\
&=
\mathop\mathrm{argmin}\limits_{x, \lambda} \left\{
\begin{aligned}&\frac{
\left\|\bar{\bm{y}}\right\|^2+
x^2\lambda^2
-2x\lambda\left|\bar{\bm{y}}^H\bar{\bm{a}}\left(\omega,\mu\right)\right|}
{x^2\left(1-\lambda^2\right)+\sigma^2} \\
&
\quad\quad\quad+\ln\left(x^2\left(1-\lambda^2\right)+\sigma^2\right)
\end{aligned}\right\}.\label{xlambda}
\end{align}
To solve \eqref{xlambda}, we make a change of the variables:
$x\rightarrow x$ and $\lambda x\rightarrow z$.
Note that within $\{x>0,0\leq\lambda\leq1\}$, this is an one-to-one map. Therefore, it is equivalent to solve
\begin{align}
\mathop\mathrm{min}\limits_{0\leq z\leq x} \quad \left\{
\begin{aligned}
&\frac{
\left\|\bar{\bm{y}}\right\|^2+
z^2
-2z\left|\bar{\bm{y}}^H\bar{\bm{a}}\left(\omega,\mu\right)\right|}
{x^2-z^2+\sigma^2}\\
&\quad\quad\quad+\ln\left(x^2-z^2+\sigma^2\right)
\end{aligned}\right\}.\label{xz}
\end{align}
We further make a change of variables:
$z\rightarrow z$ and $x^2-z^2\rightarrow t$.
Note that, in the region of $0\leq z\leq x$, this is still an one-to-one map, and therefore, \eqref{xz}
is equivalent to
\begin{align}
\mathop\mathrm{min}\limits_{t\geq0,z\geq0} \quad \frac{
\left\|\bar{\bm{y}}\right\|^2+
z^2
-2z\left|\bar{\bm{y}}^H\bar{\bm{a}}\left(\omega,\mu\right)\right|}
{t+\sigma^2}
+\ln\left(t+\sigma^2\right)\label{tz}.
\end{align}
From \eqref{tz}, it is obviously that $z\left(\omega,\mu\right)=\left|\bar{\bm{y}}^H\bar{\bm{a}}\left(\omega,\mu\right)\right|
$.
Substituting $z\left(\omega,\mu\right)$ into \eqref{tz} directly leads to
\begin{align}
&t\left(\omega,\mu\right)=\mathop\mathrm{argmin}\limits_{t>0} \quad \frac{
\left\|\bar{\bm{y}}\right\|^2-\left|\bar{\bm{y}}^H\bar{\bm{a}}\left(\omega,\mu\right)\right|^2}
{t+\sigma^2}
+\ln\left(t+\sigma^2\right)\nonumber\\
&=\left\{
\begin{matrix}
\left\|\bar{\bm{y}}\right\|^2-\left|\bar{\bm{y}}^H\bar{\bm{a}}\left(\omega,\mu\right)\right|^2-\sigma^2,&
\text{if }\Xi\left(\omega,\mu\right)\geq0,\\
0,&\text{if }\Xi\left(\omega,\mu\right)<0,\\
\end{matrix}\right. \label{Estimate:t}
\end{align}
where $\Xi\left(\omega,\mu\right)\triangleq
\left\|\bar{\bm{y}}\right\|^2-\left|\bar{\bm{y}}^H\bar{\bm{a}}\left(\omega,\mu\right)\right|^2-\sigma^2$.
Based on $z\left(\omega,\mu\right)$ and $t\left(\omega,\mu\right)$, we can obtain the expressions for $x_1\left(\omega,\mu\right)$ and $\lambda_1\left(\omega,\mu\right)$. Inserting $x_1\left(\omega,\mu\right)$, $\lambda_1\left(\omega,\mu\right)$ and $\psi\left(\omega,\mu\right)$ into \eqref{OrignalProblem}, we finally obtain that
\begin{align}
\left(\omega_1^*,\mu_1^*\right) = \mathop{\mathrm{argmax}}\limits_{-1\leq\omega\leq 1,-1\leq\mu\leq 1}\quad
\left|\bar{\bm{y}}^H\bar{\bm{a}}\left(\omega,\mu\right)\right|^2.\label{OmegaMu}
\end{align}
Substituting $\left(\omega_1^*,\mu_1^*\right)$ into $x_1\left(\omega,\mu\right)$, $\lambda_1\left(\omega,\mu\right)$, $\psi\left(\omega,\mu\right)$, we can obtain \eqref{ParameterEstimation}.
\end{proof}
Inserting \eqref{ParameterEstimation} into \eqref{GLLRInitial}, the GLLR test can be simplified as
\begin{align}
T &=\frac{\left\|\bar{\bm{y}}-\alpha_0\bar{\bm{a}}_0\right\|^2}
{\epsilon_0^2}
-\ln \frac{\left\|\bar{\bm{y}}\right\|^2 - \big(\max\limits_{\omega,\mu}\ \left|\bar{\bm{y}}^H\bar{\bm{a}}\left(\omega,\mu\right)\right|^2\big)}{\epsilon_0^2}\nonumber \\
&\gtreqless_{\mathcal{H}_0}^{\mathcal{H}_1} \tau, \label{FinalDetection}
\end{align}
where for notational simplicity, we define $\alpha_i\triangleq\sqrt{P_i}d_i^{-\alpha/2}\lambda_i$ and $\epsilon_i^2\triangleq P_id_i^{-\alpha}\delta_i^2+\sigma^2$, for $i\in\{0,1\}$.
Note that in \eqref{FinalDetection}, we only consider the cases where $\Xi^*\geq0$. We emphasize that this is a reasonable assumption because: 1) when $\lambda<1$, $\left\|\bar{\bm{y}}\right\|^2-\left|\bar{\bm{y}}^H\bar{\bm{a}}\left(\omega^*,\mu^*\right)\right|^2>0$ with probability one, and 2) the noise power, i.e., $\sigma^2$, can be significantly reduced by increasing the length of training sequence.
{ From \eqref{FinalDetection}, we can see that the computational complexity of the GLLR test mainly depends on the two-dimensional searching of $\left(\omega^*,\mu^*\right)$. For a given searching step size in each dimension, denoted by $\iota$, then the whole complexity is around $o\left(1/\iota^2\right)$.}
Now, based on \eqref{FinalDetection}, we need to design a proper decision threshold $\tau$ to ensure a satisfying authentication performance, which will be discussed in the following subsection.

\subsection{Decision Threshold Design}
\label{FARDTD}
In the proposed system, it is hard to evaluate the Bayesian decision risk because the prior knowledge about the MA is generally absent. Therefore, we resort to the Neyman-Pearson criterion \cite{M.Barkat} to determine the decision threshold. Namely,
$\tau$ is chosen such that the FAR is fixed to some value $\eta$, i.e., $\mathcal{P}_{\mathrm{FA}}\left(\tau\right)\triangleq
\mathcal{P}\left\{T>\tau|\mathcal{H}_0\right\}=\eta$.
However, the exact distribution of $T$ is very hard to obtain due to the fact it involves an exhaustive two-dimensional searching and there is no closed-form for the calculation of $T$. This means that it is even harder to obtain an analytical expression for the FAR.
To handle this problem, a computationally much more efficient approximation of the FAR is provided in the following lemma.
\begin{lemma}
\label{Lemma:EstimationFARadio}
Conditioned on $\mathcal{H}_0$, the complementary cumulative distribution function (CCDF) of $T$ in \eqref{FinalDetection} can be approximated by the CCDF of $\tilde{T}_k^{\mathrm{FA}}\triangleq X_k + Y_k - \ln Y_k$ in the sense that they have a common lower bound, where $X_k\sim\mathcal{G}\left(k,1/L\right)$, $Y_k\sim\mathcal{G}\left(L-k,1/L\right)$, and $k$ can be any integer within $\{2,3,\cdots,L-1\}$.
\end{lemma}
\begin{proof}
The basic idea for the proof is that we first find a lower bound on $T$, and then we proof that $\tilde{T}_k^{\mathrm{FA}}$ is also larger than the lower bound. Conditioned on $\mathcal{H}_0$, we have
\begin{align}
T
&\geq
\left\|\hat{\bar{\bm{g}}}_0\right\|^2
-\ln \left(\left\|\hat{\bar{\bm{y}}}\right\|^2 - \left|\hat{\bar{\bm{y}}}^H\bar{\bm{a}}_0\right|^2\right) \overset{(a)}{=}\left\|\bm{u}\right\|^2
-\ln \left\|\bm{u}_{(2,L)}\right\|^2\nonumber\\
&\leq \tilde{T}_k^{\mathrm{FA}}\triangleq \left\|\bm{u}\right\|^2
-\ln \left\|\bm{u}_{(k+1,L)}\right\|^2,\nonumber\\
&=\underbrace{\left\|\bm{u}_{(1,k)}\right\|^2}_{\triangleq X_k} + \underbrace{\left\|\bm{u}_{(k+1,L)}\right\|^2
-\ln \left\|\bm{u}_{(k+1,L)}\right\|^2}_{\triangleq Y_k - \ln Y_k},\nonumber
\end{align}
where $\hat{\bar{\bm{g}}}_0\triangleq\left(\bar{\bm{y}}-\alpha_0\bar{\bm{a}}_0\right)/\epsilon_0$, and step $(a)$ is obtain by defining $\bm{u}\triangleq\left[u_1,u_2,\cdots,u_L\right]^H=
\bm{U}_0^H\hat{\bar{\bm{g}}}_0\sim\mathcal{CN}\left(0,\frac{1}{L}\bm{I}\right)$, and $\bm{U}_0\triangleq[\bar{\bm{a}}_0,\bm{b}_{0,1},\bm{b}_{0,2},\cdots,\bm{b}_{0,L-1}]$ being a unitary matrix.
\end{proof}
Based on Lemma \ref{Lemma:EstimationFARadio}, we can use $1-\mathcal{P}\{\tilde{T}_k^{\mathrm{FA}}<\tau\}$ as an approximation of $\mathcal{P}_{\mathrm{FA}}\left(\tau\right)$, where we have
\begin{align}
&\mathcal{P}\{\tilde{T}_k^{\mathrm{FA}}<\tau\}=\frac{L^L}{\Gamma(k)\Gamma(L-k)}
\sum_{j=0}^{k-1}\binom{k-1}{j}(-1)^{k-j-1}\nonumber \\
&\quad\times\left(
\frac{\gamma\left(L-1,L;\underline{q},\overline{q}\right)}
{L-j-1}-\frac{\gamma\left(j,j+1;\underline{q},\overline{q}\right)}
{e^{(L-j-1)\tau}\left(L-j-1\right)}\right),\label{FARApproxmation}
\end{align}
with $\overline{q}$ and $\underline{q}$ denoting the roots of
$e^{x - \tau}=x$ within $[\tau,+\infty)$ and $(0,\tau]$, respectively, and $\gamma\left(n,a;b,c\right)\triangleq\int_{b}^{c}x^{n}e^{-ax}\mathrm{d}x=
\sum_{j=0}^n\frac{n!}{(n-j)!a^{j+1}}\left(b^{n-j}e^{-ab}-c^{n-j}e^{-ac}\right)$.
The derivation of $\mathcal{P}\{\tilde{T}_k^{\mathrm{FA}}<\tau\}$ is provided in Appendix.
Searching $\tilde{\tau}$ that satisfies $\mathcal{P}\{\tilde{T}_k^{\mathrm{FA}}>\tilde{\tau}\}=\eta$, we can obtain an approximation of $\tau$ that satisfies the FAR constraint.
We have to emphasize here that though $\mathcal{P}\{\tilde{T}_k^{\mathrm{FA}}>\tau\}$ is neither an upper bound nor a lower bound of the FAR, it serves as an approximation of FAR with high accuracy, which will be shown in the simulations.
\begin{figure}[!t]
	\begin{center}
		\includegraphics[width=3.5 in ]{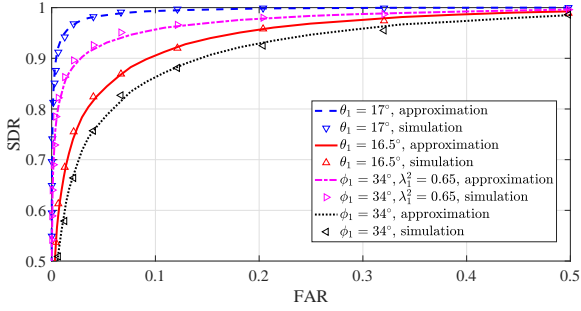}
		\caption{\small SDR v.s. FAR, where $\lambda_0^2=0.8$.}\label{Fig:Approximation}
	\end{center}
	\vspace{-5mm}
\end{figure}

\subsection{SDR Evaluation}
We now provide the performance evaluation of the above method by deriving the SDR in this subsection.
We have to point out that the exact SDR is hard to obtain due to the complicated form of $T$.
To efficiently check the authentication performance, we have the following lemma, which provides an approximation of the SDR.

\begin{lemma}
\label{LowerBoundSDR}
For a given position of the MA, conditioned on $\mathcal{H}_1$, the CCDF of $T$ can be approximated by the CCDF of $\tilde{T}^{\mathrm{SD}}\triangleq \tilde{X}+\tilde{Y}-\ln\tilde{Y}$ in the sense that they have a common lower bound, where $\tilde{Y}\sim\mathcal{G}\left(L-2,1/\varrho\right)$, and $\tilde{X}$ is a scaled non-centric chi-square random variable whose PDF is
\begin{align}
f_{\tilde{X}}\left(x\right)&=(\varrho/\left\|\bm{\beta}\right\|)\sqrt{x} e^{-\varrho\left(x+\left\|\bm{\beta}\right\|^2\right)}
I_1\left(2\varrho\sqrt{x}\right)\mathbb{I}\left\{x>0\right\},
\nonumber
\end{align}
with $I_v(x)$ being the modified bessel function of first kind with order $v$, $\mathbb{I}\{\cdot\}$ being the indicator function, $\varrho\triangleq\epsilon_0^2L/\epsilon_1^2$, $\bm{\beta}\triangleq\left[\beta_{1},\cdots,\beta_{L}\right]^T=\bm{U}_1^H\bm{\Delta}$, $\bm{U}_1\triangleq\left[\bar{\bm{a}}_1,\bm{b}_{1,1},\bm{b}_{1,2},\cdots,\bm{b}_{1,L-1}\right]$ being a unitary matrix, and $\bm{\Delta}\triangleq (1/\epsilon_0)(\alpha_1\bar{\bm{a}}_1e^{\mathrm{j}\psi} - \alpha_0\bar{\bm{a}}_0)$.
\end{lemma}
\begin{proof}
Conditioned on $\mathcal{H}_1$, we have
\begin{align}
T
&\geq
\left\|\hat{\bar{\bm{g}}}_{1}+
\bm{\Delta}\right\|^2-\ln \left(\left\|\hat{\bar{\bm{y}}}\right\|^2 - \left|\hat{\bar{\bm{y}}}^H\bar{\bm{a}}_1\right|^2\right)\nonumber\\
&\overset{(a)}{=}\left\|\bm{\beta}+\bm{v}\right\|^2
-\ln \left\|\bm{v}_{(2,L)}\right\|^2
\nonumber\\
&\overset{(b)}{=}\left|\beta_1+v_1\right|^2 + \left|\rho+z_2\right|^2 + \left\|\bm{z}_{(2,L-1)}\right\|^2
-\ln \left\|\bm{z}\right\|^2\nonumber\\
&\leq \tilde{T}^{\mathrm{SD}}\nonumber \\
&\triangleq \underbrace{\left|\beta_1+v_1\right|^2 + \left|\rho+z_2\right|^2}_{\triangleq\tilde{X}} + \underbrace{\left\|\bm{z}_{(2,L-1)}\right\|^2
-\ln\left\|\bm{z}_{(2,L-1)}\right\|^2}_{\triangleq\tilde{Y}-\ln\tilde{Y}},\nonumber
\end{align}
where $\hat{\bar{\bm{g}}}_1 \triangleq (1/\epsilon_0)(\bar{\bm{y}}-\alpha_1\bar{\bm{a}}_1e^{\mathrm{j}\psi})\sim
\mathcal{CN}\left(\bm{0},(1/\varrho)\bm{I}_L\right)$, $\hat{\bar{\bm{y}}}\triangleq(1/\epsilon_0)\bar{\bm{y}}$,
step $(a)$ is obtained by
$\bm{v}\triangleq[v_1,v_2,\cdots,v_L]^T=\bm{U}_1^H\hat{\bar{\bm{g}}}_1\sim
\mathcal{CN}\left(\bm{0},(1/\varrho)\bm{I}_L\right)$,
and step $(b)$ is obtained by
$\rho\triangleq\left\|\bm{\beta}_{(2,L)}\right\|$,
$\bm{z}\triangleq[z_2,z_3,\cdots,z_L]^T= \check{\bm{U}}_1^H\bm{v}_{(2,L)}\sim
\mathcal{CN}\left(\bm{0},(1/\varrho)\bm{I}_{L-1}\right)$,
and $\check{\bm{U}}_1 \triangleq [\bm{\beta}_{(2,L)}/\left\|\bm{\beta}_{(2,L)}\right\|, \bm{c}_2,\cdots,\bm{c}_{L-1}]$ being a unitary matrix.
\end{proof}

\begin{figure*}
\begin{align}
&\mathcal{P}\left\{\tilde{T}^{\mathrm{SD}}>\tau\right\}=1-\iint\limits_{x+y-\ln y<\tau} f_{\tilde{X}}(x)f_{\tilde{Y}}(y)\mathrm{d}y\mathrm{d}x=
1-\int_{\underline{q}}^{\overline{q}} \int_{e^{u-\tau}}^u f_{\tilde{X}}(u-v)f_{\tilde{Y}}(v)\mathrm{d}v\mathrm{d}u\nonumber\\
&=1-\frac{\varrho^2}{\left\|\bm{\beta}\right\|}e^{-\varrho\left\|\bm{\beta}\right\|^2}
\frac{\varrho^{L-2}}{\Gamma(L-2)}
\sum_{k=0}^{+\infty}\frac{\varrho ^{2k}}{k!\Gamma(k+2)}\int_{\underline{q}}^{\overline{q}}\int_{e^{u-\tau}}^u
(u-v)^{k+1}v^{L-3}e^{-\varrho u}
\mathrm{d}v\mathrm{d}u\nonumber\\
&=1 - \frac{\varrho^2}{\left\|\bm{\beta}\right\|}e^{-\varrho\left\|\bm{\beta}\right\|^2}
\frac{\varrho^{L-2}}{\Gamma(L-2)}
\sum_{k=0}^{+\infty}\frac{\varrho ^{2k}}{k!\Gamma(k+2)}\sum_{j=0}^{k+1}\binom{k+1}{j}(-1)^{k+1-j}\nonumber\\
&\quad\quad\quad\times
\frac{\gamma(k+L-1,\varrho; \underline{q}, \overline{q}) - e^{-\left(k-j+L-1\right)\tau}\gamma(j,\varrho+j+1-k-L; \underline{q}, \overline{q})}{k-j+L-1}.
\label{SDRApproximation}
\end{align}
\noindent\rule[-1mm]{\textwidth}{1pt}
\end{figure*}
Based on Theorem \ref{LowerBoundSDR}, we can use $\mathcal{P}\{\tilde{T}^{\mathrm{SD}}>\tau\}$ to approximate the SDR. The calculation of $\mathcal{P}\{\tilde{T}^{\mathrm{SD}}>\tau\}$ is provided in \eqref{SDRApproximation} on the top of next page. The derivation of \eqref{SDRApproximation} is quite similar to that of \eqref{FARApproxmation} in Appendix, and thus the detailed steps are omitted. In the simulation part, we will show that the proposed approximation in \eqref{SDRApproximation} is very accurate to evaluate the SDR.
\begin{figure*}[t]
\begin{align}
&\mathcal{P}\left\{\tilde{T}_k^{\mathrm{FA}}<\tau\right\}=\mathcal{P}\left\{X_k+Y_k-\ln Y_k<\tau\right\}\overset{(a)}{=}\int_{\underline{q}}^{\overline{q}}\int_{e^{u-\tau}}^u
f_{X_k,Y_k}(u-v,v)\mathrm{d}v\mathrm{d}u\nonumber\\
&\overset{(b)}{=}\frac{L^L}{\Gamma(k)\Gamma(L-k)}
\sum_{j=0}^{k-1}\binom{k-1}{j}(-1)^{k-j-1}
\int_{\underline{q}}^{\overline{q}}
u^je^{-Lu}
\int_{e^{u-\tau}}^uv^{L-j-2}\mathrm{d}v\mathrm{d}u\overset{(c)}{=}\mathrm{Eqn.}~\eqref{FARApproxmation}.
\label{DerivationofFAR}
\end{align}
\noindent\rule[-1mm]{\textwidth}{1pt}
\end{figure*}

\section{Numerical Results}
\label{Sec:NumResults}
\begin{figure}[t]
\begin{center}
\subfigure[SDR v.s. $\theta_1$, where $\lambda_0^2=0.8$.]{
    \label{dJA:sub1} 
    \includegraphics[width=1.62 in]{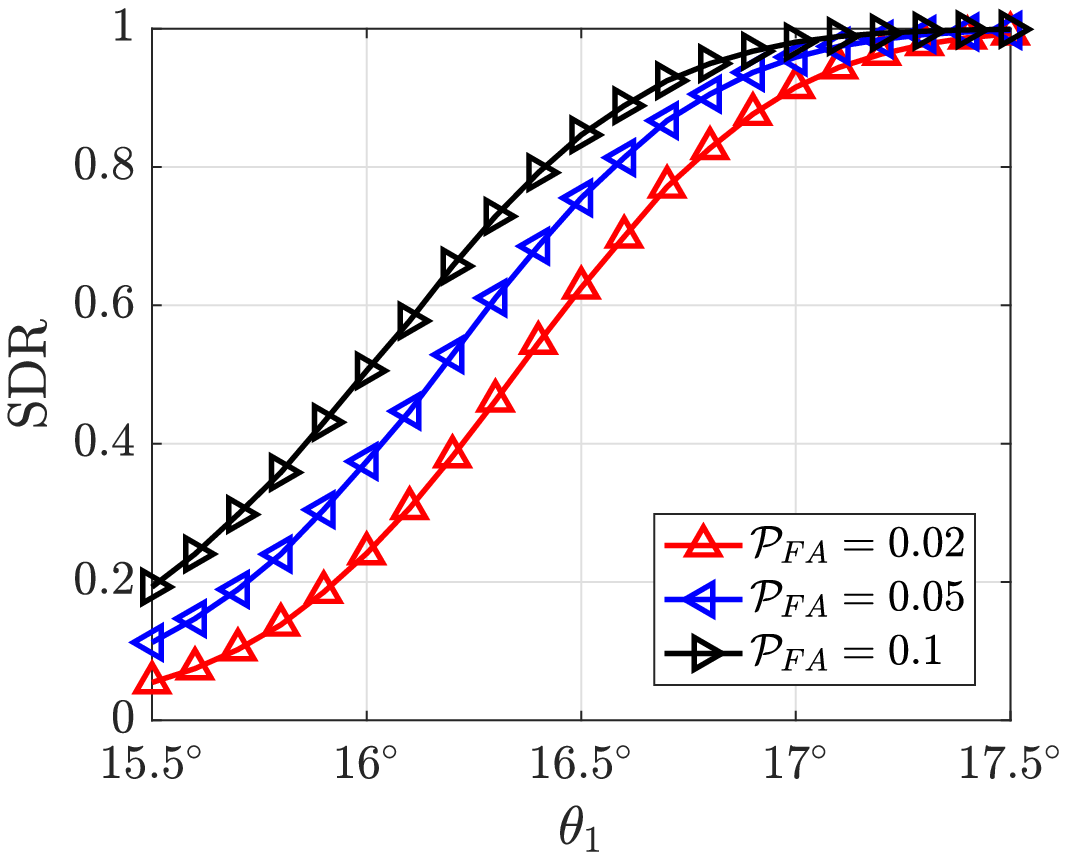}}
    \hspace{.01in}
\subfigure[SDR v.s. $\phi_1$, where $\lambda_0^2=0.8$.]{
    \label{dJA:sub2} 
    \includegraphics[width=1.62 in]{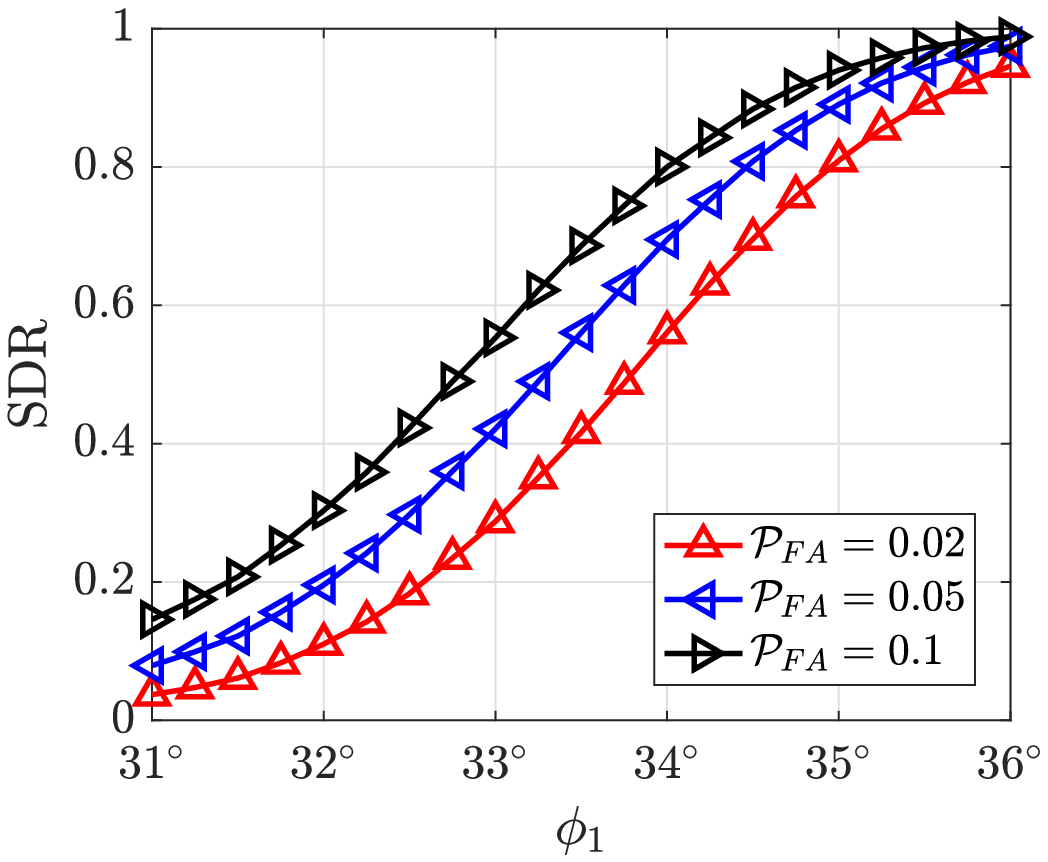}}
    \vspace{.01in}
\subfigure[SDR v.s. $\lambda_1^2$, where $\lambda_0^2=0.85$ and $\theta_1=16^{\circ}$.]{
    \label{dJA:sub3} 
    \includegraphics[width=1.62 in]{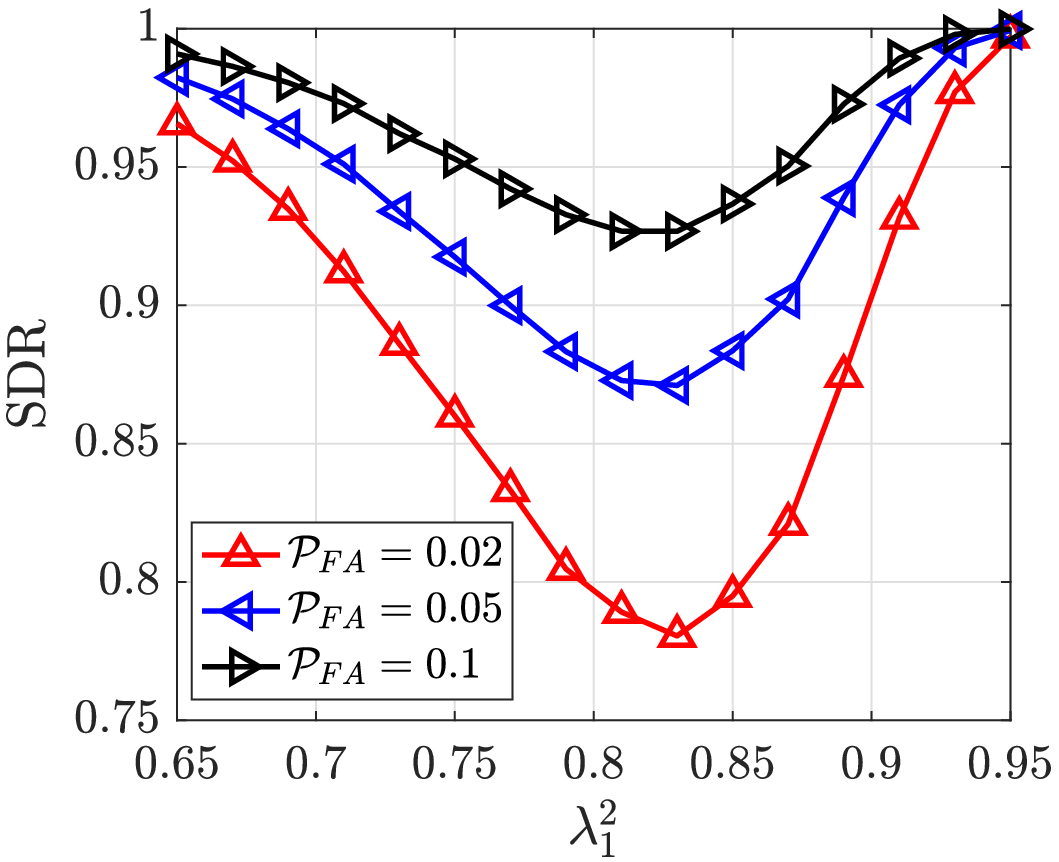}}
    \hspace{.01in}
\subfigure[SDR v.s. $P_1$, where $\lambda_0^2=0.85$ and $\phi_1=34^{\circ}$.]{
    \label{dJA:sub4} 
    \includegraphics[width=1.62 in]{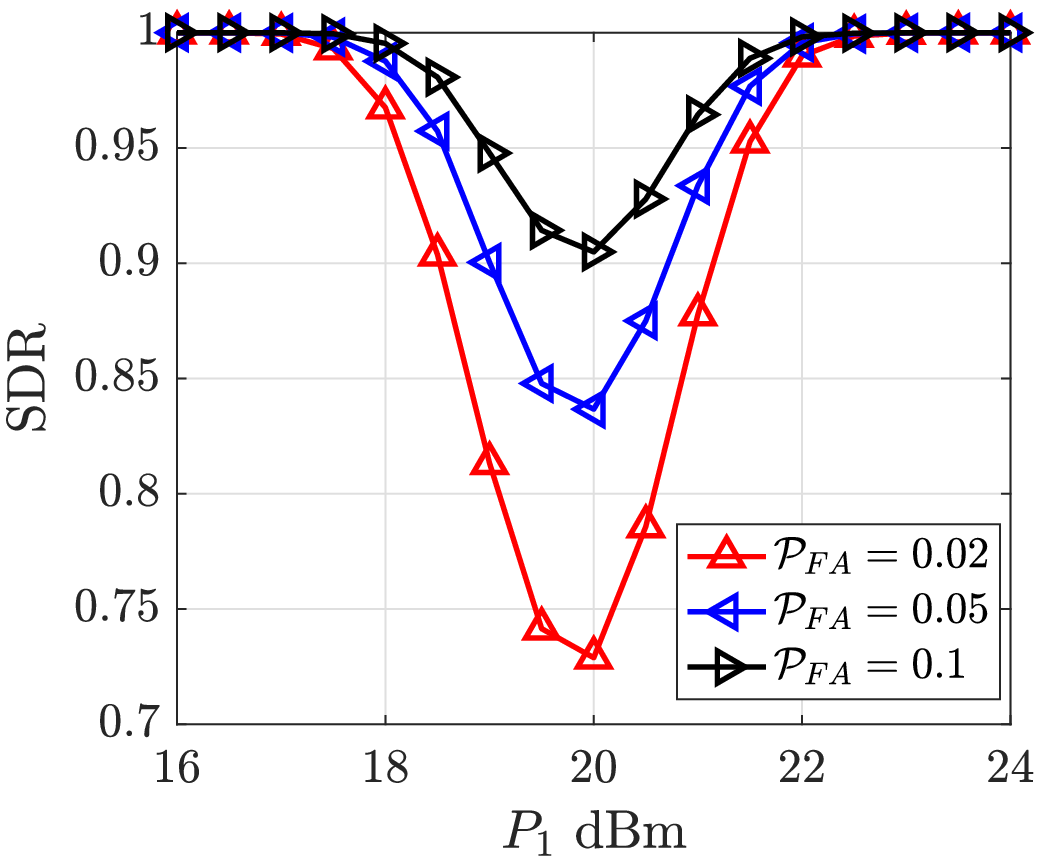}}
\caption{\small SDR under fixed FAR.}\label{Fig:dJA}
\end{center}
\vspace{-5mm}
\end{figure}
In this section, we provide some numerical results to show the performance of the proposed authentication  method. We assume the UAV is equipped with a T-shaped array, with $2M+1$ antennas placed along $x$-axis and $N$ antennas placed along $y$-axis. We assume the adjacent antennas are separated by half of a wavelength. Then, we have
$\bm{a}\left(\omega_i,\mu_i\right)=
[\bm{a}_X\left(\omega_i\right)^T,\bm{a}_Y\left(\mu_i\right)^T]^T$
where $\bm{a}_X\left(x\right) = [e^{\mathrm{j}M\pi x},\cdots,1,\cdots,
e^{-\mathrm{j}M\pi x}]^T$, and
$\bm{a}_Y\left(x\right) = [e^{-\mathrm{j}\pi x},\cdots,
e^{-\mathrm{j}N\pi x}]^T$.
In the simulation, we set $\sigma^2=0.01$ mW, $M=6$, $N=12$, and the UAV is at the height of $20$ m. Under $\mathcal{H}_0$, we set $\theta_0=15^\circ$,$\phi_0=30^\circ$,$P_0=20$ dBm. Under $\mathcal{H}_1$, the corresponding parameters are the same as those  under $\mathcal{H}_0$ unless specified.
The estimations of $\omega_1$ and $\mu_1$ in \eqref{OmigaMuEstimation} are obtained by extensively searching with step size given by $0.005$.

In Fig. \ref{Fig:Approximation}, the receiver operating characteristic curves are plotted under some given positions of the MA. The approximations of the SDR and the FAR are obtained by using the CCDF of $\tilde{T}^{\mathrm{SD}}$ and $\tilde{T}_2^{\mathrm{FA}}$, respectively. As we can see, the approximation results can accurately approximate the simulation results.

In Fig. \ref{Fig:dJA}, we plot the SDRs against some key physical layer parameters of the MA, i.e., $\theta_1$, $\phi_1$, $\lambda_1^2$, and $P_1$, while the FARs are fixed as $0.02$, $0.05$, and $0.1$.
{ In Fig. \ref{dJA:sub1} and Fig. \ref{dJA:sub2}, we show that with
$\left(\theta_1,\phi_1\right)$ deviating from $\left(\theta_0,\phi_0\right)$, the SDRs get improved. This is because the spoofing signal comes from an undesired direction which leads to the mismatch of the estimated and the desired directions of arrival. However, if $\left(\theta_1,\phi_1\right)$ closely approaches $\left(\theta_0,\phi_0\right)$, the SDRs decline sharply due to the limited spatial resolution.
In Fig. \ref{dJA:sub3}, we show that for a given location of the MA, the SDRs increase if $\lambda_1$ gets distinct from $\lambda_0$. In fact, the value of $\lambda_i$, for $i\in\{0,1\}$,  influences the strength of the LOS and the non-LOS components of the corresponding physical channel. When $\lambda_1$ is significantly distinct from $\lambda_0$, then the ratio between the LOS and non-LOS signal powers of the spoofing control signal from the MA will not coincide with that of the real control signal from the GCS.
In practice, the Rician-$\kappa$ factors are different at different locations, which can be utilized to improve the authentication performance. In Fig. \ref{dJA:sub4}, we plot the SDRs versus the transmit powers of the MA. In the simulation, the distances from the UAV to the GCA and the MA are the same. It can be observed that the lowest SDR appears under the cases when $P_0 = P_1$. This is because the deviation of $P_1$ from $P_0$ will increase or decrease the total received power at the UAV and thus becomes more possible to be successfully detected.
}

\section{Conclusion}
\label{Sec:Conclusion}

In this correspondence,
we focused on the authentication at the UAV-side by considering a MA transmitting forged control signal to pretend as the legitimate GCS.
We considered the worst case where the UAV has no prior knowledge of the MA.
We proposed a GLLR-based authentication method in this paper.
The exact FAR and SDR were hard to obtain due to the complicated form of the GLLR. To evaluate the authentication performance, accurate approximations were provided.

\appendix
The derivation of $\mathcal{P}\{\tilde{T}_k^{\mathrm{FA}}<\tau\}$ is provided in \eqref{DerivationofFAR} on the top of this page, where step $(a)$ is obtained by $x+y\rightarrow u$ and $y\rightarrow v$, and step $(b)$ is obtained by using
binomial expansion, and step $(c)$ is obtained by first integrating with respect to $v$ and then following the definition of $\gamma(n,a; b,c)$.

\end{document}